\documentclass[12pt,a4paper]{article}
\usepackage{geometry}
\usepackage{amsmath,amssymb}
\usepackage{amsfonts}
\usepackage{amsthm}
\usepackage{mathrsfs}
\numberwithin{equation}{section}
\usepackage{comment}
\newcommand{\jap}[1]{\langle #1 \rangle}
\def\a{\alpha}
\def\b{\beta}
\def\c{\gamma}

\def\e{\varepsilon}
\def\f{\varphi}
\def\g{\psi}

\def\m{\mu}

\def\x{\xi}
\def\y{\eta}

\newcommand{\Op}{\mathrm{Op}}
\def\re{\mathbb{R}}

\def\pa{\partial}

\renewcommand{\Re}{\text{{\rm Re}\;}}
\renewcommand{\Im}{\text{{\rm Im}\;}}

\newcommand{\supp}{\text{{\rm supp}\;}}

\newtheorem{thm}{Theorem}[section]
\newtheorem{lem}[thm]{Lemma}
\newtheorem{prop}[thm]{Proposition}

\theoremstyle{definition}
\newtheorem{defn}{Definition}

\theoremstyle{remark}

\title{Remarks on the geodesically completeness and the smoothing effect on asymptotically Minkowski spacetimes}

\author{Kouichi Taira\thanks{e-mail:ktaira@fc.ritsumei.ac.jp}}

\date{}

\begin{document}

\maketitle

\begin{abstract}
In this note, we study a geometric property of asymptotically Minkowski spacetimes and an analytic property of the Klein-Gordon operator. Precisely, our first main results show that asymptotically Minkowski spacetimes are geodesically complete under a null non-trapping condition. Secondly, we prove that Sobolev index of a real principal type estimate used in the previous work is actually optimal.
\end{abstract}

\section{Introduction}

Let $g_0$ be a Minkowski metric on $\re^{n+1}$ and $g_0^{-1}$ be its dual metric:
\begin{align*}
g_0=-dx_1^2+dx_2^2+...+dx_{n+1}^2,\quad g_0^{-1}=-\pa_{x_1}^2+\pa_{x_2}^2+...+\pa_{x_{n+1}}^2=(g^{ij}_0)_{i,j=1}^n,
\end{align*}
where we denote
\begin{align*}
x=(t,y)\in \re\times \re^n,
\end{align*}
that is, $t$ is the time variable and $y$ is the space variable in the spacetime.
We denote $\jap{x}=(1+|x|^2)^{\frac{1}{2}}$ and introduce the function space
\begin{align*}
S^k(\re^{n+1}):=\{a\in C^{\infty}(\re^{n+1})\mid |\pa_x^{\a}a(x)|\leq C\jap{x}^{k-|\a|}\},\quad \quad k\in \re.
\end{align*}

\begin{defn}
A Lorentzian metric $g$ on $\re^{n+1}$ is called asymptotically Minkowski if the inverse matrix $g^{-1}(x)=(g^{jk}(x))_{j,k=1}^n$ of $g(x)$ satisfies $g^{jk}- g_0^{jk}\in S^{-\m}(\re^{n+1})$ for some $\m>0$.

\end{defn}

The Feynman propagator, which is an inverse of $-\Box_g$ and satisfies a certain wavefront condition, is a fundamental object in Quantum field theory. In \cite{DeZ1}, \cite{DeZ2}, \cite{GW1}, \cite{GW2} and \cite{Va1}, the Feynman propagator is constructed on various spacetimes including asymptotically Minkowski spacetimes. In \cite{T}, it is proved that on asymptotically Minkowski spacetimes, the (anti-)Feynman propagator constructed in \cite{GW1} and \cite{GW2} coincides with the outgoing resolvent of the Klein-Gordon operator. Such identity should hold since for the exact Minkowski spacetime, the Feynman propagator is defined by the outgoing resolvent in physics literatures. See also the review article \cite{GW3}. 

In this short note, we give supplementary results on the geometry and the property of the Klein-Gordon operator on asymptotically Minkowski spacetimes: One is a result on the completeness of asymptotically Minkowski spacetimes and the other is a result on the optimality of the local smoothing estimate used in \cite{T}.

\subsection{First result}

It is a classical question how the essential self-adjointness of the pseudodifferential operator is related to the completeness of the associated Hamilton flow (\cite[Vol II]{RS}, \cite{BMS}). In the case of the Laplace operator $P$ on a semi-Riemannian manifold $M$, this corresponds to a relation between the essential self-adjointness of $P$ and the geodesically completeness of $M$. It is well-known that on the geodesically complete Riemannian manifod, the Laplace operator is essentially self-adjoint on $C_c^{\infty}$. In \cite{BMS} and references therein, such a relation is studied for more general Schr\"odinger type operators on complete Riemmanian manifolds.

In \cite{CB}, it is shown that the completeness of the Hamilton flow is equivalent to the essential self-adjointness on generic closed Lorentzian surfaces. The authors in \cite{CB} conjecture that the completeness of the Hamilton flow implies the essential self-adjointness for non-elliptic operators on closed manifold. In \cite{T}, this conjecture is solved for general real principal type oeprators on one-dimensional torus.
As a related result, in \cite{K}, the author gives an example of a Lorentzian manifold $(M,g)$ which is geodesically complete and globally hyperbolic, but the Laplace operator associated with $g$ is not essentially self-adjoint. 

Recently, it is shown in \cite{Va1} and \cite{NT} that the wave operator on an asymptotically Minkowski spacetime is essentially self-adjoint on $C_c^{\infty}$ under a null non-trapping condition. However, on this manifold, a relation to the geodesically completeness was not revealed. In this note, we show that a null non-trapping condition implies the completeness on this manifold, including timelike and spacelike completeness.

We introduce a non-trapping condition which is a bit weaker than the condition in \cite{GW1}, \cite{GW2}, \cite{NT} and \cite{T} in the sense that the completeness of the null geodesics are not assumed here.

\begin{defn}\cite[Definition 11.17]{B}\label{assnulltra}
Let $(M,g)$ be a Lorentzian manifold. We say that a maximally extended geodesic is forward (resp. backward) non-trapping if for each $t_0\in (-T_0,T_1)$, $\c|_{(t_0,T_1)}$ (resp.  $\c|_{(-T_0,t_0] }$) fails to have compact closure. We way that $\c$ is non-trapping if $\c$ is both forward and backward non-trapping. The Lorentzian manifold $(M,g)$ is called null non-trapping (or null disprisoning) if all non-constant maximally extended null geodesic is non-trapping. 

\end{defn}

\begin{thm}\label{thmcomplete}
Suppose $(\re^{n+1},g)$ is asymptotically Minkowski and null non-trapping in the sense of Definition $\ref{assnulltra}$. Then $(\re^{n+1},g)$ is geodesically complete.
\end{thm}

\subsection{Second result}

Next, we consider the exact Minkowski spacetime $(\re^{n+1},g_0)$ and the operator
\begin{align*}
P=\pa_t^2-\Delta_y.
\end{align*}
It is well-known that $P$ is essentially self-adjoint on $C_c^{\infty}(\re^{n+1})$. We denote the unique self-adjoint extension of $P$ by its same symbol. In \cite[Proposition 3.2]{T}, it is proved that the resolvent $(P-i)^{-1}$ has the following mapping property:
\begin{align}\label{locsm1}
(P-i)^{-1}:L^2(\re^{n+1})\to \jap{x}^{\frac{1}{2}+\e}H^{\frac{1}{2}}(\re^{n+1}),\quad \e>0.
\end{align}
This mapping property plays a crucial role for the proof of the limiting absorption principle in \cite{T} on asymptotically Minkowski spacetimes. On the other hand, the radial estimate (\cite[Propositions A.3, A.4]{DW}, \cite[Theorem A.3]{T}, \cite[(2), (3)]{Va1}) and the propagation of singularities imply
\begin{align}\label{locsm2}
(P-i)^{-1}:\jap{x}^{-\frac{1}{2}-\e}L^2(\re^{n+1})\to \jap{x}^{\frac{1}{2}+\e} H^1(\re^{n+1}),\quad \e>0.
\end{align}
Namely, $(\ref{locsm2})$ states that if $f$ has an additional decay (compared with the case $(\ref{locsm1})$), then its image $(P-i)^{-1}f$ gains a better regularity.
So a natural question is whether the Sobolev index $\frac{1}{2}$ in $(\ref{locsm1})$ is optimal or not. In this note, we give an affirmative answer to this question.

\begin{thm}\label{thmoptimal}
For the Minkowski spacetime $(\re^{n+1},g_0)$, for all $\e>0$, there exists $f\in L^2(\re^{n+1})$ such that $(P-i)^{-1}f\notin H^{\frac{1}{2}+\e}_{loc}(\re^{n+1})$.
\end{thm}

Usual real principal type estimates are used for proving existence of a local solution (\cite[Theorem 26.1.7]{Ho}). The estimates like $(\ref{locsm1})$ and $(\ref{locsm2})$ might be regarded as global versions of the real principal type estimates. Theorem $(\ref{thmoptimal})$ states that a global analoge of the usual real principal type estimates with Sobolev index $1$ does not hold if we impose $f\in L^2(\re^{n+1})$ only.
We note that Sobolev index $\frac{1}{2}$ in $(\ref{locsm1})$ is same as the local smoothing effects for time-dependent Schr\"odinger equations (see also \cite[\S 1.2]{T}).

\subsection{Organization and notation}

This paper is organized as follows: In Section \ref{sectiongeodesic}, the proof of the geodesically completeness is given. In Section \ref{sectionopti}, we show the optimality of the mapping property $(\ref{locsm1})$.
Moreover, in the Appendix, we give a short proof of $(\ref{locsm1})$ and $(\ref{locsm2})$ on the exact Minkowski spacetime.

%\subsection*{Notation}
We fix some notations. We denote $\jap{x}=(1+|x|^2)^{1/2}$ for $x\in \re^{n+1}$. We use the Sobolev spaces: $H^k(\re^{n+1})=\jap{D}^{-k}L^2(\re^{n+1})$ for $k\in \re$. We denote the Fourier transform on $\re^m$ by $\mathcal{F}f(y)=(2\pi)^{-\frac{m}{2}}\int_{\re^m}e^{-iy\cdot\y}f(y)dy$. For a Banach space $X$, we denote the norm of $X$ by $\|\cdot\|_{X}$. If $X$ is a Hilbert space, we write the inner metric of $X$ by $(\cdot, \cdot)_{X}$, where $(\cdot, \cdot)_{X}$ is linear with respect to the right variable. For a Lorentzian manifold $(M,g)$, we say that $v\in T_pM$ is a null vector if $g_p(v,v)=0$.

\section{Geodesically completeness of the Asymptotically Minkowski spacetimes}\label{sectiongeodesic}

\subsection{Completeness of trapped geodesics}

In this subsection, we discuss completeness of non-null trapped geodesics.
The following lemma is possibly well-known, however, the author cannot find a suitable reference. Hence we give a proof here. 

\begin{lem}\label{generalcompl}
Suppose a manifold $M$ satisfies the second axiom of countability. Suppose that a Lorentzian manifold $(M,g)$ is null non-trapping and $\c:(a,b)\to M$ is a maximally extended forward $($resp. backward$)$ trapped non-null geodesic, that is, $g(\c'(t),\c'(t))\neq 0$ and $\c([0,b))\subset K$ $($resp. $\c((a,0])\subset K$$)$ with a compact set $K\subset M$. Then $\c$ is forward complete $($resp. backward complete$)$ in the sense that $b=\infty$ $($resp. $a=\infty$$)$.
\end{lem}

\begin{proof}
We consider the forward case only.
Since $M$ satisfies the second axiom of countability, there exists a complete Riemannian metric on $M$ (see \cite{NO}). We fix a complete Riemannian metric $h$ on $M$. Let
\begin{align*}
SM=\{(x,v)\in TM\mid h(v,v)=1 \}
\end{align*}
be the sphere bundle associated with $h$. Since $\c$ is trapped, there exists $b\in \re$ such that $\c([0,b))\subset K$ for a compact set $K\subset M$. 

Set 
\begin{align*}
v(t)=\frac{1}{h(\c'(t),\c'(t))^{\frac{1}{2}}}  \c'(t)\in S_{\c(t)}M.
\end{align*}
Let $\f_t$ be the geodesic flow on $TM$ and $\pi:TM\to M$ be the natural projection.
Since $(\c(t),v(t))\in SM|_K:=SM\cap \pi^{-1}(K)$ and $SM|_K$ is compact, there exist a sequence $\{t_i\}_i$ and $(q,v)\in SM|_K$ such that $t_i\to b$ and $(\c(t_i),v(t_i))\to (q,v)\in SM$ as $i\to \infty$.  We shall show that $v$ is not null. We suppose $v$ is a null vector.
Then the null non-trapping condition implies the existence of $T>0$ such that $\f_t(q,v)$ exists for $0\leq t\leq T$ and $\f_T(q,v)\notin SM|_K$. By the continuous dependence of initial datas of ODE, we have $\f_T(\c(t_i),v(t_i))\to \f_T(q,v)$ as $i\to \infty$. In particular, $\f_T(\c(t_i),v(t_i))\notin SM|_K$. This contradicts to 
\begin{align*}
\f_T(\c(t_i),v(t_i))=\f_{\frac{T}{h(\c'(t),\c'(t))^{\frac{1}{2}}}} (\c(t_i),\c'(t_i))\in SM|_K.
\end{align*}
Thus it follows that $v$ is a not null vector.

Suppose $\c$ is not forward complete. Then we have $b<\infty$ and 
\begin{align*}
0\neq g(v,v)=\lim_{i\to\infty} g(v(t_i),v(t_i))=\lim_{i\to \infty}\frac{g(\c'(t_i),\c'(t_i))}{h(\c'(t_i),\c'(t_i))}.
\end{align*}
Notice that $g(\c'(t_i),\c'(t_i))$ is non zero constant. Thus we have
\begin{align*}
\lim_{i\to\infty}h(\c'(t_i),\c'(t_i))=\frac{g(\c'(t_i),\c'(t_i))}{g(v,v)}<\infty
\end{align*}
Since $\c([0,b))\subset K$ and $h$ is complete Riemannian metric, it contradicts to $b<\infty$ and to \cite[Proposition 2.1]{G}.
\end{proof}

\subsection{Completeness of non-trapping orbits}

In the following of this section, we assume $g$ is an asymptotically Minkowski metric on $\re^{n+1}$. We set
$p(x,\x)=\frac{1}{2}\sum_{j,k=1}^ng^{ij}(x)\x_i\x_j$. Let $(y(t,x,\x),\y(t,x,\x))$ denote the solution to the Hamilton equation:
\begin{align}\label{Hameq}
\begin{cases}
\frac{d}{dt}y(t,x,\x)=\pa_{\x}p(y(t,x,\x),\y(t,x,\x)),\\
\frac{d}{dt}\y(t,x,\x)=-\pa_{x}p(y(t,x,\x),\y(t,x,\x)),
\end{cases}\quad  \begin{cases}
y(0,x,\x)=x,\\
\y(0,x,\x)=\x.
\end{cases}
\end{align}
It is well-known that $t\mapsto y(t,x,\x)$ is a geodesic on the Lorentzian manifold $(M,g)$ with the initial value $x$ and the initial velocity $g(x)^{-1}\x$.
The following lemma follows from a direct calculation. See \cite[Lemma A.1]{NT} for a proof.
\begin{lem}\label{clmourre}
There exist $M>0$ and $R_0>1$ such that 
\[
H_{p}^2(|x|^2)\geq M|\x|^{2}
\]
for any $(x,\x)\in \{(y,\y)\in T^*\re^{n}\,|\, |y|>R_0,\,|\y|\neq 0\}$. 
\end{lem}

We also mention a result on an extension of solutions to the Hamilton equation.

\begin{lem}\label{solext}
Let $(x,\x)\in T^*\re^{n+1}$. Suppose that the solution $(y(t,x,\x),\y(t,x,\x))$ to $(\ref{Hameq})$ exists for a time interval $(T_0,T_1)$ and that there exists $C>0$ such that
\begin{align*}
|\y(t,x,\x)|\leq C\quad \text{for}\quad t\in (T_0,T_1).
\end{align*}
Then the solution $(y(t,x,\x),\y(t,x,\x))$ can be extended to a time interval beyond $(T_0,T_1)$.
\end{lem}

\begin{proof}
We note $|\pa_{\x}p(x,\x)|\leq C'|\x|$ with a constant $C'>0$.
By the assumption and the Hamilton equation, we have $|y(t,x,\x)|\leq |x|+CC'|t|$. Thus $(y(t,x,\x),\y(t,x,\x))$ stays a fixed compact set for $t\in (T_0,T_1)$. A standard theory of ODE gives our conclusion.
\end{proof}

Now we show that every non-trapping orbits are complete on the asymptotically Minkowski spacetiems.
The following proposition is proved by a slight modification of the proof in \cite{NT}, which follows the strategy in \cite{KPRV}.

\begin{prop}\label{nontrappcomplete}
Fix $(x_0,\x_0)\in T^{*}\re^n$ with $\x_0\neq 0$ and suppose that $(x_0,\x_0)$ is forward $($resp. backward$)$ non-trapping in the sense that there exists $T\in (0,\infty]$ $($resp. $[-\infty.0)$$)$ such that 
\begin{align*}
\lim_{t\to T,\, t<T}|y(t,x_0,\x_0)|=\infty\quad  (\text{resp.}\,\, \lim_{t\to T,\, t>T}|y(t,x_0,\x_0)|=\infty)
\end{align*}
Then, there exist $C_1, C_2>0$ such that
\[
C_1\leq |\y(t, x_0, \x_0)|\leq C_2\quad \text{for}\quad 0\leq t< T\quad (\text{resp.} -T< t\leq 0) .
\]
Moreover, it follows that $|T|=\infty$ and that the orbit $(y(t,x_0,\x_0),\y(t,x_0,\x_0))$ is forward $($resp. backward$)$ complete.
\end{prop}

\begin{proof}
We consider the forward case only.
We write $y(t)=y(t,x_0,\x_0)$ and $\y(t)=\y(t,x_0,\x_0)$.
Let $R_0$ be as in Lemma~\ref{clmourre}, and we let $R_1\geq R_0$ which is determined later.
We first note that by the forward non-trapping condition and Lemma \ref{clmourre}, there exits $0\leq t_0<T$ such that for $t_0\leq t<T$, we have
\begin{equation}\label{nontrap}
|y(t)|\geq R_1,\quad \frac{d}{dt}|y(t)|^2\geq 0.
\end{equation}
Indeed, it is easy to see that there are $0<s_0<t_0<T$ such that 
$\frac{d^2}{dt^2}|y(t)|^2>0$ for $t\geq s_0$, and $\frac{d}{dt}|y(t_0)|^2> 0$. 
Then for all $t_0\leq t<T$, the condition \eqref{nontrap} is satisfied. 

Take a constant $C_0>0$ such that
\begin{align*}
|\pa_xp(x,\x)|\leq C_0|x|^{-1-\m}|\x|^2\quad \text{for}\quad |x|\geq 1.
\end{align*}
We write $\y_0=|\y(t_0)|$ and $T_1:=\sup\{s\in [t_0,T)\mid  \frac{1}{2}\y_0\leq |\y(t)|\}$.
By Lemma \ref{clmourre} and $(\ref{nontrap})$, we have
\begin{align*}
|y(t)|^2\geq R_1^2+\frac{M\y_0^2}{8}(t-t_0)^2\quad \text{for}\quad t_0\leq t <T_1.
\end{align*}
Since $R_1\geq 1$, the Hamilton equation gives $|\y'(t)|\leq C_0|y(t)|^{-1-\m}|\y(t)|^2$ for $t_0\leq t <T_1$
and hence
\begin{align*}
|\frac{d}{dt}|\y(t)|^{-1}|\leq C_0\left(R_1^2+\frac{M\y_0^2}{8}(t-t_0)^2\right)^{-\frac{1+\m}{2}}\quad \text{for}\quad t_0\leq t <T_1.
\end{align*}
Thus, for $t_0\leq t <T_1$, we obtain
\begin{align*}
|\y_0^{-1}-|\y(t)|^{-1}|\leq& C_0\int_{t_0}^t\left(R_1^2+\frac{M\y_0^2}{8}(s-t_0)^2\right)^{-\frac{1+\m}{2}}ds\leq \frac{2\sqrt{2}C_0C_{\m}}{R_1^{\m}M^{\frac{1}{2}}}\y_0^{-1}
\end{align*}
where $C_{\m}=\int_0^{\infty}(1+s^2)^{-\frac{1+\m}{2}}ds$. Taking $R_1$ large enough, we have $\frac{2}{3}\y_0\leq |\y(t)|\leq \frac{4}{3}\y_0$ for $t_0\leq t<T_1$. By the definition of $T_1$, we obtain $T_1=T$. Moreover, by Lemma \ref{solext} and the inequality $\frac{2}{3}\y_0\leq |\y(t)|\leq \frac{4}{3}\y_0$ for $t_0\leq t<T=T_1$, we conclude $T=\infty$.

\end{proof}

\subsection{Proof of the first main result}

Now we prove our first main result.

\begin{proof}[Proof of Theorem \ref{thmcomplete}]
By Proposition \ref{nontrappcomplete}, any non-trapping geodesics are complete. Due to the null non-trapping condition, it suffices to prove that all trapped non-null geodesics are complete. The completeness of trapped non-null geodesics follow from Lemma \ref{generalcompl}. This completes the proof.
\end{proof}

\section{Optimality for the local smoothing estimate}\label{sectionopti}

In this section, we assume that $g=g_0$ is the Minkowski metric on $\re^{n+1}=\re_t\times \re^n_y$.

\subsection{An explicit formula of the resolvent $(P-i)^{-1}$}

We recall $P=\pa_t^2-\Delta_y$ on $\re^{n+1}$. First, we calculate an explicit formula of the resolvent $(P-i)^{-1}$. A similar formula also appears in the proof of \cite[Theorem C.1]{DW}. To do this, at first, we shall calculate the imaginary part of the symbol of $\sqrt{-\Delta_y-i}$. Set
\begin{align*}
A:=\sqrt{-\Delta_y+i}=A_1-iA_2\quad A_1,A_2:\text{self-adjoint}
\end{align*}
with $A_2\geq 0$.

\begin{lem}\label{imagcal}
For $\y\in \re^n$, set $a=\sqrt{|\y|^2-i}=re^{i\theta}$ with $\Im a\leq 0$. Then we have
\begin{align*}
r=(|\y|^4+1)^{\frac{1}{4}},\quad \Im a=-\frac{1}{2^{\frac{1}{2}}(\sqrt{|\y|^4+1}+|\y|^2)^{\frac{1}{2}}}.
\end{align*}
In particular, $C^{-1}\jap{\y}\leq|a|\leq C\jap{\y}$ and $C^{-1}\jap{\y}^{-1}\leq|\Im a|\leq C\jap{\y}^{-1}$ with a constant $C>0$.
\end{lem}

\begin{proof}
Since $a^2=|\y|^2-i$, we have $r^2=|a|^2=|a^2|=||\y|^2-i|=(|\y|^4+1)^{\frac{1}{2}}$, which implies $r=(|\y|^4+1)^{\frac{1}{4}}$. Moreover, it follows that $r^2\cos2\theta+ir^2\sin2\theta= r^2e^{2i\theta}=|\y|^2-i$ and  
\begin{align*}
\cos2\theta=\frac{|\y|^2}{(|\y|^4+1)^{\frac{1}{2}}},\quad \sin2\theta=\frac{-1}{(|\y|^4+1)^{\frac{1}{2}}},\quad \sin^2\theta=\frac{1}{2(|\y|^4+1)^{\frac{1}{2}}((|\y|^4+1)^{\frac{1}{2}}+|\y|^2 ) }.
\end{align*}
This calculation with $\Im a\leq 0$ implies
\begin{align*}
\Im a=r\sin\theta=-\frac{1}{2^{\frac{1}{2}}((|\y|^4+1)^{\frac{1}{2}}+|\y|^2 )^{\frac{1}{2}} }.
\end{align*}

\end{proof}

We write
\begin{align*}
a_1=\Re a,\quad a_2=-\Im a,\quad A_1=a_1(D_{\y}),\quad A_2=a_2(D_{\y}). 
\end{align*}
Since $C^{-1}\jap{\y}^{-1}\leq|\Im a|\leq C\jap{\y}^{-1}$, we have
\begin{align}\label{Soboleveq}
C^{-1}\|u\|_{H^{k}(\re^{n}_y)}\leq \|A_2^{-k}u\|_{L^2(\re^{n}_y)} \leq C\|u\|_{H^{k}(\re^{n}_y)}
\end{align}
for all $k\in \re$ with a constant $C>0$. 

Now we calculate an expression of the resolvent $(P-i)^{-1}$.

\begin{prop}
For $f\in L^2(\re^{n+1})$, we have
\begin{align}
(P-i)^{-1}f(t,y)=&\frac{i}{2}A^{-1}\int_{-\infty}^te^{-i(t-s)A}f(s,y)ds-\frac{i}{2}A^{-1}\int_{+\infty}^te^{i(t-s)A}f(s,y)ds\nonumber\\
=&\frac{i}{2}A^{-1}\int_{-\infty}^{\infty}e^{-i|t-s|A}f(s,y)ds.\label{resformula}
\end{align}

\end{prop}

\begin{proof}
Set $Rf(t,y)=\frac{i}{2}A^{-1}\int_{-\infty}^{\infty}e^{-i|t-s|A}f(s,y)ds$. Since $P$ is essentially self-adjoint on $C_c^{\infty}(\re^{n+1})$, it suffices to prove that $(P-i)Rf=f$ for $f\in C_c^{\infty}(\re^{n+1})$ and that $R$ is a bounded operator on $L^2(\re^{n+1})$. The identity  $(P-i)Rf=f$ follows from a direct calculation. Hence we shall prove that $R$ is bounded on $L^2(\re^{n+1})$. Set $a_2(\y)=-\Im a$. Using the Fourier transform in the $y$-variable and a scaling, we have
\begin{align*}
\|Rf\|_{L^2(\re^{n+1}_{t,y})}=&\|\mathcal{F}_{y\to \y}Rf\|_{L^2(\re^{n+1}_{t,\y})}\leq \frac{1}{2}\|a^{-1}\int_{\re}e^{-|t-s|a_2}|\mathcal{F}_{y\to \y}f(s,\y)|ds\|_{L^2(\re^{n+1}_{t,\y})}\\
\leq&\frac{1}{2}\|a_2^{-\frac{3}{2}}a^{-1}\int_{\re}e^{-|t-s|}|\mathcal{F}_{y\to \y}f(\frac{s}{a_2},\y)|ds\|_{L^2(\re^{n+1}_{t,\y})}\\
\leq&\frac{1}{2}\|e^{-|t|}\|_{L^1(\re_t)} \|a_2^{-\frac{3}{2}}a^{-1}\mathcal{F}_{y\to \y}f(\frac{t}{a_2},\y)\|_{L^2(\re^{n+1}_{t,\y})}\\
=&\frac{1}{2}\|e^{-|t|}\|_{L^1(\re_t)} \|a_2^{-1}a^{-1}\mathcal{F}_{y\to \y}f(t,\y)\|_{L^2(\re^{n+1}_{t,\y})}\leq C\|f\|_{L^2(\re^{n+1})}.
\end{align*}
Here we use the Young inequality in the third line and set 
\begin{align*}
C=\frac{\|e^{-|t|}\|_{L^1(\re_t)}}{2}\sup_{\y\in \re^n}|a_2(\y)^{-1}a(\y)^{-1}|,
\end{align*}
which is finite by virtue of Lemma \ref{imagcal}. This completes the proof.
\end{proof}

\subsection{Proof of the second result}

In this subsection, we shall prove Theorem \ref{thmoptimal}. To do this, it suffices to find $f\in L^2(\re^{n+1})$ such that
\begin{align*}
\jap{D_{\y}}^{\frac{1}{2}+\e}u  \notin L^2_{loc}(\re^{n+1}).
\end{align*}
where we set $u:=(P-i)^{-1}f$ and $\y\in \re^n$ is the dual variable of $y$.

Let $\chi\in C_c^{\infty}((\frac{1}{4},1):[0,1])$ and $\g\in C_c^{\infty}(\re^n;[0,1])$ satisfying $\|\chi\|_{L^2(\re)}=\|\g\|_{L^2(\re^n)}=1$ and
\begin{align*}
\chi(t)=1\quad  \text{on}\quad \frac{1}{2}\leq t\leq 1,\quad \g(\y)=1\quad \text{on}\quad \frac{1}{2}\leq |\y|\leq 1.
\end{align*}
We denote the Fourier transform from the variable $\y$ to the variable $y$ by $\mathcal{F}^{-1}_{\y\to y}$ and set
\begin{align*}
g(t,\y):=\jap{\y}^{-\frac{n+1}{2}-\e}\chi(\frac{t}{\jap{\y}})e^{ita_1(\y)}\in L^2(\re^{n+1}),\quad f(t,y)= \mathcal{F}^{-1}_{\y\to y}g(t,y), \quad u:=(P-i)^{-1}f.
\end{align*}
We shall show that
\begin{align}\label{optfunction}
\f(y)\jap{D_{\y}}^{\frac{1}{2}+\e}u\notin  L^2((0,\frac{1}{4})_t\times \re^n_y).
\end{align} 
with $\f\in C_c^{\infty}(\re^n_y)$ which is determined as follows: Let $M>0$ be a constant satisfying $|\pa_{\y}a_1(\y)|\leq M$ for all $\y\in \re^n$. Take $\f,\f_1\in C_c^{\infty}(\re^n;[0,1])$ such that $\f_1(y)=\f_1(-y)$ and
\begin{align*}
\f(y)=1\quad \text{on}\quad |y|\leq M,\quad \f_1(y)=1\quad \text{on}\quad |y|\leq \frac{M}{8},\quad \supp \f_1\subset \{|y|\leq \frac{M}{4}\}.
\end{align*}
As is seen below, $e^{-itA_1}u$ has no oscillation terms and is easy to handle. In order to remove the oscillation factor $e^{itA_1}$, we will use Egorov's theorem and an elliptic estimate in the microlocal analysis. Precisely, we shall prove
\begin{align}
\f_1\jap{D_{\y}}^{\frac{1}{2}+\e}e^{-itA_1}u&\notin L^2((0,\frac{1}{4})_t\times \re^n_y),\label{red1}\\
\|\f_1\jap{D_{\y}}^{\frac{1}{2}+\e}e^{-itA_1}u\|_{ L^2((0,\frac{1}{4})_t\times \re^n_y)} &\leq C\|\f\jap{D_{\y}}^{\frac{1}{2}+\e}u\|_{ L^2((0,\frac{1}{4})_t\times \re^n_y)}+C\|u\|_{L^2(\re^{n+1})}\label{red2}
\end{align}
which ensure $(\ref{optfunction})$. Thus it remains to prove $(\ref{red1})$ and $(\ref{red2})$.

\subsection*{Proof of $(\ref{red1})$}

First, we deal with $(\ref{red1})$.
 Since $t-s\leq 0$ for $t\leq \frac{1}{4}$ and $s\in \supp \chi(\frac{\cdot}{\jap{\y}})$, it follows from  $(\ref{resformula})$ that for $t\leq \frac{1}{4}$, we have
\begin{align*}
\mathcal{F}_{y\to \y}(\jap{D_{\y}}^{\frac{1}{2}+\e}e^{-itA_1}u)(t,\y)=&\frac{i}{2}\jap{\y}^{-\frac{n}{2}} a(\y)^{-1}\int_t^{+\infty}e^{-|t-s|a_2(\y)}\chi(\frac{s}{\jap{\y}})ds\\
=&\frac{i}{2}\jap{\y}^{-\frac{n}{2}+1} a(\y)^{-1}\int_{\jap{\y}^{-1}t}^{+\infty}e^{-|\jap{\y}^{-1}t-s|\jap{\y}a_2(\y)}\chi(s)ds=:b_t(t,\y).
\end{align*}
Thus we have $\f_1(y)(\jap{D_{\y}}^{\frac{1}{2}+\e}e^{-itA_1}u)(t,y)=\f_1(y)(\mathcal{F}_{\y\to y}^{-1}b_t)(y)$. By Lemma \ref{imagcal}, $b_t$ satisfies
\begin{align}\label{b_tes}
|\pa_{\y}^{\a}b_t(\y)|\leq C_{\a}\jap{\y}^{-\frac{n}{2}},\quad C^{-1}\jap{\y}^{-\frac{n}{2}}\leq |b_t(\y)|\leq C\jap{\y}^{-\frac{n}{2}}
\end{align}
uniformly in $t\in (0,\frac{1}{4})$.

\begin{lem}
Let $b_t\in C^{\infty}(\re^n)$ satisfying $(\ref{b_tes})$ uniformly in $t$.
Then we have $\f_1\mathcal{F}_{\y\to y}^{-1}b_t\notin L^2(\re^n_y)$ for each $t$.
\end{lem}

\begin{proof}
By the Plancherel theorem, we have $\mathcal{F}_{\y\to y}^{-1}b_t\notin L^2(\re^n_y)$. Thus it suffice to prove $(1-\f_1)\mathcal{F}_{\y\to y}^{-1}b_t\in L^2(\re^n_y)$. By integrating by parts, it turns out that $\mathcal{F}_{\y\to y}^{-1}b_t(y)$ is rapidly decreasing and smooth away from $y=0$. Since $\f_1(y)=1$ near $y=0$, then we obtain $(1-\f_1)\mathcal{F}_{\y\to y}^{-1}b_t\in L^2(\re^n_y)$.
\end{proof}

Now we suppose $\f_1(y)\jap{D_{\y}}^{\frac{1}{2}+\e}e^{-itA_1}u\in  L^2((0,\frac{1}{4})_t\times \re^n_y)$. Then we have 
\begin{align*}
\f_1(y)\jap{D_{\y}}^{\frac{1}{2}+\e}e^{-itA_1}u\in  L^2(\re^n_y)
\end{align*}
for almost everywhere $t\in (0,\frac{1}{4})$. This contradicts to the lemma above. We complete the proof of $(\ref{red1})$.

\subsection*{Proof of $(\ref{red2})$}

Next, we prove $(\ref{red2})$.
We briefly recall the notion of pseudodifferential operators. For precise treatment, see \cite[\S 18]{Ho}. We define the Weyl quantization of a symbol $a$ by
\begin{align*}
\mathrm{Op}(a)u(y)=\frac{1}{(2\pi)^{n}}\int_{\re^{2n}}e^{i(y-y')\cdot \y}a(\frac{y+y'}{2},\x)u(y')dy'd\y.
\end{align*}
We note that if $a$ is real-valued, then $\Op(a)$ is formally self-adjoint and if $a\in S^0$, then $\Op(a)$ is bounded in $L^2(\re^n)$.
For $k\in \re$ and $a\in C^{\infty}(\re^{2n})$, we call $a\in S^{k}$ if
\begin{align*}
|\pa_{y}^{\a}\pa_{\y}^{\b}a(x,\x)|\leq C_{\a\b}\jap{\y}^{k-|\b|}.
\end{align*}
We define $\{a,b\}:=H_ab:=\pa_{\y}a\cdot \pa_{y}b-\pa_{y}a\cdot \pa_{\y}b$.
For $a\in S^{k_1}$ and $b\in S^{k_2}$, we have
\begin{align}\label{compfor}
\Op(a)\Op(b)=\Op(ab)+\Op S^{k_1+k_2-1},\,\, [\Op(a),i\Op(b)]=\Op(\{a,b\})+\Op S^{k_1+k_2-2}
\end{align}
We need the following lemmas.

\begin{lem}[Egorov's theorem]\label{Ego}
Set $a_t(y,\y):=\f(y-t\pa_{\y}a_1(\y))$. Then
\begin{align*}
e^{-itA_1}\f(y)e^{itA_1}+R_t=\Op(a_t),
\end{align*}
where $R_t\in B(H^{-1}(\re^n), L^2(\re^n))$ is an operator locally uniformly bounded in $t$.

\end{lem}

\begin{proof}
Using $(e^{itA_1}\Op(a_t)e^{-itA_1})'|_{t=0}=\f(y)$, we have $\Op(a_t)=e^{-itA_1}\f(y)e^{itA_1}+R_t$, where
\begin{align*}
R_t=\int_0^te^{-i(t-s)A_1}(\frac{d}{ds}\Op(a_s)+[A_1,i\Op(a_s)])e^{i(t-s)A_1}ds=\int_0^te^{-i(t-s)A_1}L_se^{i(t-s)A_1}ds
\end{align*}
and $L_s=\frac{d}{ds}\Op(a_s)+[A_1,i\Op(a_s)]$. The formula $(\ref{compfor})$ implies that $L_s\in \Op S^{-1}$ is locally uniformly bounded in $t$. Since $e^{-itA_1}$ preserves $H^k(\re^n)$, it follows that $R_t\in B(H^{-1}(\re^n), L^2(\re^n))$ is locally uniformly bounded in $t$.
\end{proof}

\begin{lem}[Elliptic parametrix]\label{cutell}
For $0\leq t\leq \frac{1}{4}$, we have
\begin{align}\label{ellsupport}
(\supp\f_1)\times \re^n\subset \{(y,\y)\in \re^{2n}\mid \f(y-t\pa_{\y}a_1(\y))=1 \},
\end{align}
that is, $\f(y-t\pa_{\y}a_1(\y))$ is elliptic on $\supp\f \times\re^n$ in the phase space. Moreover, there exist $c_t\in S^{0}$ and $r_t\in S^{-1}$ which are locally uniformly bounded in $t$ such that
\begin{align}\label{ellcomp}
\f_1(y)=\Op(c_t)\Op(a_t)+\Op(r_t)
\end{align}
where we recall $a_t(y,\y):=\f(y-t\pa_{\y}a_1(\y))$.
\end{lem}

\begin{proof}
For $0\leq t\leq \frac{1}{4}$ and $|y|\leq \frac{M}{4}$, we have
\begin{align*}
|y-t\pa_{\y}a(\y)|\leq |y|+|t||\pa_{\y}a(\y)|\leq \frac{M}{4}+\frac{1}{4}\cdot M< M.
\end{align*}
which implies $(\ref{ellsupport})$. We note that the smooth function $c_t(y,\y)=\f_1(y)/a_t(y.\y)$ is well-defined by the support condition $(\ref{ellsupport})$. Since $a_t\in S^0$ locally uniformly bounded in $t$, $c_t$ is also locally uniformly bounded in $t$. The composition formula $(\ref{compfor})$ gives $(\ref{ellcomp})$ with a symbol $r_t\in S^{-1}$ locally uniformly bounded in $t$.
\end{proof}

Now we prove $(\ref{red2})$. Set
\begin{align*}
C_1:=\sup_{t\in (0,\frac{1}{4})}\max(\|\Op(c_t)\|_{L^2(\re^n)\to L^2(\re^n)},\|R_t\|_{H^{-\frac{1}{2}-\e}(\re^n)\to L^2(\re^n)}, \|\Op(r_t)\|_{H^{-\frac{1}{2}-\e}(\re^n)\to L^2(\re^n)} ).
\end{align*}
For $t\in (0,\frac{1}{4})$, by Lemmas \ref{Ego} and \ref{cutell}, we obtain
\begin{align*}
\|\f_1\jap{D_{\y}}^{\frac{1}{2}+\e}e^{-itA_1}u\|_{ L^2(\re^n_y)} \leq& C_1\|\Op(a_t)\jap{D_{\y}}^{\frac{1}{2}+\e}e^{-itA_1}u\|_{ L^2(\re^n_y)}+C_1\|u\|_{L^2(\re^n)}\\
\leq&\|\f\jap{D_{\y}}^{\frac{1}{2}+\e}u\|_{L^2(\re^n)}+(C_1^2+C_1)\|u\|_{L^2(\re^n)},
\end{align*}
where we use the unitarity of $e^{-itA_1}$. Integrating this inequality in $t$, we obtain $(\ref{red2})$.

\appendix

\section{A short proof for smoothing effects on the Minkowski spacetimes}

In this appendix, we give a short proof for $(\ref{locsm1})$ and $(\ref{locsm2})$ using the explicit formula $(\ref{resformula})$.

\begin{lem}\label{lemapp}
Set
\begin{align*}
I=\int_{\re}e^{-i|t-s|A}f(s)ds.
\end{align*}
Then, for $\e>0$ we have
\begin{align*}
\|\jap{t}^{-\frac{1}{2}-\e}I\|_{L^2(\re^{n+1})}\leq C\|\jap{D_y}^{\frac{1}{2}}f\|_{L^2(\re^{n+1})},\quad \|\jap{t}^{-\frac{1}{2}-\e}I\|_{L^2(\re^{n+1})}\leq C\|\jap{t}^{\frac{1}{2}+\e}f\|_{L^2(\re^{n+1})}
\end{align*}

\end{lem}

\begin{proof}
The second inequality immediately follows from H\"older's inequality. Thus we shall show the first inequality. We write $\hat{f}(t,\y)=\mathcal{F}_{y\to \y}f(t,\y)$ and $a_2:=-\Im a=-\Im \sqrt{|\y|^2-i}$, where $\mathcal{F}_{y\to\y}$ denotes the Fourier transform from the variable $y$ to the variable $\y$.
Fourier transforming in $y\to \y$ and using Young's inequality, we have
\begin{align*}
|\hat{I}(t,\y)|\leq |\int_{-\infty}^t  e^{-(t-s)a_2}|\hat{f}(s)|ds|\leq \|e^{-ta_2}\|_{L^2_t} \|\hat{f}(t,\y)\|_{L^2_t}\leq C\|a_2^{-\frac{1}{2}}\hat{f}(t,\y)\|_{L^2_t}
\end{align*}
This calculation gives $\|\jap{t}^{-\frac{1}{2}-\e}\hat{I}(t,\y)\|_{L^2_t}\leq \|\jap{t}^{-\frac{1}{2}-\e}\|_{L^2_t}\|\hat{I}(t,\y)\|_{L^{\infty}_t}\leq C\|a_2^{-\frac{1}{2}}\hat{f}(t,\y)\|_{L^2_t}$. Plancherel's theorem and $(\ref{Soboleveq})$ imply
\begin{align*}
\|\jap{t}^{-\frac{1}{2}-\e}I\|_{L^2(\re^{n+1})}\leq C\|a_2^{-\frac{1}{2}}\hat{f}\|_{L^2(\re^{n+1})}\leq C\|\jap{D_y}^{\frac{1}{2}}f\|_{L^2_tL^2_y}.
\end{align*}
\end{proof}

This lemma and the formula $(\ref{resformula})$ immediately imply
\begin{align}\label{locsmoy}
\|\jap{t}^{-\frac{1}{2}-\e}(P-i)^{-1}f\|_{L^2(\re^{n+1})}\leq C\|A^{-1}\jap{D_y}^{\frac{1}{2}}f\|_{L^2(\re^{n+1})}\leq C\|\jap{D_y}^{-\frac{1}{2}}f\|_{L^2(\re^{n+1})}.
\end{align}
Let $\g\in C_c^{\infty}(\re;[0,1])$ such that $\g(s)=1$ on $|s|\leq \frac{1}{4}$ and $\supp g\subset \{|s|\leq \frac{1}{2}\}$. Set $\f(D_{x})=\g(P/(-\Delta_x+1))$. Since $P$ is elliptic on the essential support of $1-\f(D_{x})$, we have $\|(1-\f(D_x))\jap{D_x}^{\frac{1}{2}}(P-i)^{-1}f\|_{L^2(\re^{n+1})}\leq C\|f\|_{L^2(\re^{n+1})}$. Moreover, since $D_{t}\sim D_y$ on the essential support of $\f(D_x)$, we have 
\begin{align*}
\|\jap{t}^{-\frac{1}{2}-\e}\f(D_x)\jap{D_x}^{\frac{1}{2}}(P-i)^{-1}f\|_{L^2(\re^{n+1})}\leq C\|\jap{t}^{-\frac{1}{2}-\e}\f(D_x)\jap{D_y}^{\frac{1}{2}}(P-i)^{-1}f\|_{L^2(\re^{n+1})}.
\end{align*}
Combining these inequalities with $(\ref{locsmoy})$, we obtain $(\ref{locsm1})$.

The mapping property $(\ref{locsm2})$ immediately follows from the formulae $(\ref{resformula})$, 
\begin{align*}%\label{formulat}
\pa_t(P-i)^{-1}f(t,y)=\frac{1}{2}\int_{-\infty}^te^{-i(t-s)A}f(s,y)ds+\frac{1}{2}\int_{+\infty}^te^{i(t-s)A}f(s,y)ds,
\end{align*}
and the second inequality in Lemma \ref{lemapp}. See also the proof of \cite[Theorem C.1]{DW}.


\begin{thebibliography}{99}


\bibitem{B} J. K. Beem, P. E. Ehrlich, K. L. Easley, Global Lorentzian geometry. Second edition. Monographs and Textbooks in Pure and Applied Mathematics, 202. Marcel Dekker, Inc, New York, 1996.


\bibitem{BMS} M. Braverman, O. Milatovich, M. Shubin, Essential self-adjointness of Schr\"odinger-type operators on manifolds, Russian Math. Surveys 57 (4) (2002) 641--692.


\bibitem{CB} Yves Colin de Verdi\'ere, C. Bihan, On essential-selfadjointness of differential operators on closed manifolds, to appear in Ann. Fac. Sci. Toulouse Math, arXiv:2004.06937.

\bibitem{DW} D. N. Dang, M. Wrochna, Complex powers of the wave operator and the spectral action on Lorentzian scattering spaces, preprint, (2020), arXiv:2012.00712.

\bibitem{DeZ1} J. Derezi\'nski, D. Siemssen, Feynman propagators on static spacetimes, Rev. Math. Phys. 30, (2018), 1850006.

\bibitem{DeZ2} J. Derezi\'nski, D. Siemssen, An evolution equation approach to the Klein-Gordon operator on curved spacetime, Pure Appl. Anal. 1, 215--261, (2019).


\bibitem{DS} J. Derezi\'nski, D. Siemssen, An Evolution Equation Approach to Linear Quantum Field Theory, preprint, (2019), arXiv:1912.10692.

\bibitem{G} M. S\'anchez, An Introduction to the Completeness of Compact Semi-Riemannian Manifolds, S\'emin. Th\'eor. Spectr. G\'eom., 13, Univ. Grenoble I, Saint-Martind’H\'eres, 1995, 37--53.

\bibitem{GW1} C. G\'erard, M. Wrochna, The massive Feynman propagator on asymptotically Minkowski spacetimes, Amer. J. Math. \textbf{141}, (2019), 1501--1546.


\bibitem{GW2} C. G\'erard, M. Wrochna, The massive Feynman propagator on asymptotically Minkowski spacetimes II, Int. Math. Res. Notices. \textbf{2020}, (2020), 6856--6870.

\bibitem{GW3} C. G\'erard and M. Wrochna. The Feynman problem for the Klein-Gordon equation. arXiv:2003.14404, 2020.

\bibitem{Ho} L. H\"ormander,  Analysis of Linear Partial Differential Operators, Vol.\ I-IV.  Springer Verlag, 1983--1985. 

\bibitem{K} W. Kami\'nski, Non self-adjointness of the Klein-Gordon operator on globally hyperbolic and geodesically complete manifold. An example, preprint, arXiv:1904.03691.

\bibitem{KPRV} C. Kenig, G. Ponce, C. Rolvung, L. Vega: Variable coefficient Schr\"odinger flows for ultrahyperbolic operators. Adv. Math. Vol. 196, No. 2, pp. 373--486 (2005) .


\bibitem{NO} K. Nomizu, H. Ozeki, The existence of complete Riemannian metrics, Proc. Amer. Math. Soc. \textbf{12} (1961), 889--891.

\bibitem{NT} S. Nakamura, K. Taira, Essential self-adjointness of real principal type operators, to appear in Ann. Henri Lebesgue, arXiv:1912.05711.

\bibitem{RS} M. Reed, B. Simon, The Methods of Modern Mathematical Physics, Vol.\ I--IV.  Academic Press, 1972--1980. 

\bibitem{T2} K. Taira, Equivalence of classical and quantum completeness for real principal type operators on the circle, arXiv:2004.07547.

\bibitem{T} K. Taira, Limiting absorption principle and equivalence of Feynman propagators on asymptotically Minkowski spacetimes, arXiv:2012.07946.


\bibitem{Va1} A. Vasy: Essential self-adjointness of the wave operator and the limiting absorption principle on Lorentzian scattering spaces. J. Spectr. Theory. \textbf{10}, (2020), 439--461.






\end{thebibliography}
\end{document}